\documentclass[11pt]{article}
\usepackage[margin=1.1in]{geometry}

\usepackage{amsfonts}   
\usepackage{amsmath}    
\usepackage{amsthm}     
\usepackage{mathrsfs}   
\usepackage{mathtools}  

\newtheorem{define}{Definition}[section]
\newtheorem{lemma}[define]{Lemma}

\newtheorem{thm}[define]{Theorem}

\newtheorem{conj}[define]{Conjecture}

\usepackage{hyperref}
\hypersetup{
  linkcolor=blue,        
  citecolor=blue,        
  colorlinks=true,       
  linktocpage=true,
}

\newcommand{\del}{\backslash}
\newcommand{\con}{/}
\newcommand{\minor}{\subseteq_\mathcal{M}}
\newcommand{\rw}{\text{rw}}
\newcommand{\cw}{\text{cw}}

\begin{document}

\author{
  Nikola Yolov\\
  \texttt{nikola.yolov@cs.ox.ac.uk}\\
  Department of Computer Science\\
  University of Oxford
}
\title{Minor-matching hypertree width}

\maketitle

\abstract {
  In this paper we present a new width measure for a tree decomposition,
  \emph{minor-matching hypertree width}, $\mu\text{-}tw$,
  for graphs and hypergraphs,
  such that bounding the width guarantees
  that set of maximal independent sets has a polynomially-sized
  restriction to each decomposition bag.
  The relaxed conditions of the decomposition allow
  a much wider class of graphs and hypergraphs of bounded width
  compared to other tree decompositions.
  We show that, for fixed $k$, there are $2^{(1 - \frac1k + o(1)){n \choose 2}}$
  $n$-vertex graphs of minor-matching hypertree width at most $k$.
  A number of problems including
  Maximum Independence Set,
  k-Colouring, and
  Homomorphism of uniform hypergraphs
  permit polynomial-time solutions
  for hypergraphs
  with bounded minor-matching hypertree width
  and bounded rank.
  We show that
  for any given $k$ and any graph $G$,
  it is possible to
  construct a decomposition of minor-matching hypertree width at most $O(k^3)$,
  or to prove that $\mu\text{-}tw(G) > k$
  in time $n^{O(k^3)}$.
  This is done by presenting a general algorithm for approximating
  the hypertree width of well-behaved measures,
  and reducing $\mu\text{-}tw$ to such measure.
  The result relating the restriction of the maximal independent sets to a set $S$
  with the set of induced matchings intersecting $S$ in graphs,
  and minor matchings intersecting $S$ in hypergraphs,
  might be of independent interest.
}

\section{Introduction}

In his classical paper~\cite{tarjan-clique-decomp},
Tarjan discusses divide and conquer algorithms
for solving problems on chordal graphs using
a tree decomposition where every bag is a clique.
Since chordality is independent of the much more widely used
notion of treewidth~\cite{graph_minors_3},
it is a natural question to unify the two
tree decompositions into a single one.
For instance, we may consider a tree decomposition of a graph,
with width measured as the maximum independence number of a bag,
say $\alpha\text{-}tw(G)$, instead of the maximum number of vertices in a bag less
$1$, $tw(G)$.
This indeed unifies the two decompositions,
as a graph $G$ is chordal if and only if $\alpha\text{-}tw(G) = 1$
and clearly $\alpha\text{-}tw(G) \le tw(G) + 1$.

However, it turns out that $\alpha\text{-}tw(G)$ is often not general enough.
For instance, consider the \textbf{$k$-Colouring} problem.
It is possible to determine if $\chi(G) \le k$
in polynomial time for graphs with bounded $\alpha\text{-}tw(G)$,
since there are polynomially many distinct colourings of each bag,
we may enumerate all of them and apply dynamic programming.
Nevertheless,
the class of graphs $\{G : \alpha\text{-}tw(G) \le k\}$
is not a new polynomial case of the problem,
because $tw(G) \le \alpha\text{-}tw(G)\chi(G)$.
However, with little more effort we can describe
a much more relaxed tree decomposition measure.

The trace (or restriction) of a set system $\mathcal{F}$ over a set $S$,
denoted $tr_S(\mathcal{F})$ is defined as $\{F \cap S : F \in \mathcal{F}\}$.
Let us denote the set of all independent and all \emph{maximal} independent sets
of a graph $G$ by $\hat{i}(G)$ and $i(G)$ respectively.
Solving problems on graphs with $\alpha\text{-}tw(G) \le k$
exploits the property that $tr_B(\hat{i}(G))$ has polynomially-bounded
size for each bag $B$ of a tree decomposition of $G$
with bounded $\alpha$-width,
but it often suffices to require that $tr_B(i(G))$ has bounded size.
Let $\mu(G)$ denote the maximum size of an induced matching in $G$.
Results in graph theory~\cite{farber1, farber2, alekseev}
show that the relation between $|\hat{i}(G)|$ and $\alpha(G)$
is similar to the relation between $|i(G)|$ and $\mu(G)$:
$2^{\alpha(G)} \le |\hat{i}(G)| \le {v(G) \choose \le \alpha(G)}$
and $2^{\mu(G)} \le |i(G)| \le {e(G) \choose \le \mu(G)}$.
Now let the \emph{$S$-intersecting induced matching number}, $\mu_G(S)$,
be the maximum size of an induced matching $M$ of $G$
such that every edge of $M$ intersects $S \subseteq V(G)$.
Clearly $\alpha(G[S]) \ge \mu_G(S)$,
since taking a vertex from every edge of an induced matching
forms an independent set.
On the contrary, $\mu_G(S)$ cannot be bounded from below using $\alpha(G[S])$,
since $\alpha(E_n) = n$, while $\mu(E_n) = 1$,
where $E_n$ is the $n$-vertex empty graph.
In this paper we show that
$tr_S(i(G))$ is polynomially bounded if $\mu_G(S)$ is bounded.
This result is of independent interest,
but in particular it motivates the parameter
that is discussed in this paper, $\mu\text{-}tw(G)$,
a width measure of graphs,
defined as the minimum over all tree decompositions of $G$
of the maximum intersecting induced matching number
over the bags of the decomposition.

Hypergraphs often have much more expressive power for
modelling problems compared to graphs.
All results above generalise to hypergraphs $H$,
but the definition of $\mu_H(S)$ is too technical to be defined here,
so we postpone it to \S~2.

When comparing different tree decompositions,
usually the following four questions are considered.
The first one is \emph{on what structures does the decomposition work}.
The original treewidth decomposition is defined on graphs,
but later work introduces tree decompositions for hypergraphs,
for example Hypertreewidth~\cite{hypertreefirst}
and Fractional hypertree width decompositions~\cite{ftwfirst}.
The decomposition in this paper yields efficient algorithms
for \emph{hypergraphs with hyperedges of bounded size}.
It is not clear whether or not the restriction on the size
of the hyperedges is necessary, but so far we can only guarantee
polynomial-time solutions when the condition holds.

The second question is
\emph{what is the number of structures with bounded width}.
Here minor-matching hypertree decompositions really stand out,
as most graphs and hypergraphs have much lower
minor-matching hypertree width compared to other width measures.
We show that, for fixed $k$, there are $2^{(1 - \frac1k + o(1)){n \choose 2}}$
$n$-vertex graphs with minor-matching hypertree width at most $k$
(the number of all $n$-vertex graphs is $2^{n \choose 2}$).
In comparison,
most tree decompositions permit $2^{O(n \ln n)}$ $n$-vertex graphs
of width at most $k$.
More detail is given at the end of $\S$~2.

The next question is
\emph{
  what family of problems can be solved in polynomial time
  provided a decomposition of the instance with bounded width
  is given as input}.
For the moment there is no meta-theorem similar
to Courcelle's theorem~\cite{courcelle_thm_1},
describing classes of computable problems,
and instead we give examples how to solve
\textbf{Maximum Weighted Independent Set},
\textbf{$k$-Colouring} and
\textbf{Homomorphism of Uniform Hypergraphs} in $\S$~2.
Other problems, like \textbf{Set Cover} and \textbf{k-SAT},
also have natural solutions using the developed methods.

The final question is \emph{
  is it possible to efficiently construct a tree decomposition of width
  $k$, or more generally $f(k)$, given a structure of width at most $k$.
}
This problem is essential, as it allows us to drop
the requirement to provide a tree decomposition in the input
of the problem, and find truly polynomial cases of hard problems.
A tree decomposition without this property has limited use,
since it may be as hard to find a decomposition
as to solve the problem.
Given a fixed $k$ and a problem $P$ that permits a polynomial
solution of instances $(I, T)$,
where $T$ is a decomposition of width at most $k$,
it suffices to construct decomposition of width $f(k)$
to solve $P(I)$, as $f(k)$ is also bounded from above.
In $\S$~4 we present an $n^{O(k^3)}$-time algorithm that given
a graph $G$ either finds a tree decomposition of $G$ with $\mu$-width
$O(k^3)$,
or proves that $\mu\text{-}tw(G) > k$.
An intermediate step is a similar result for $\alpha\text{-}tw(G)$ of
hypergraphs.

The paper is organised as follows.
In $\S$~2 we present minor-matching hypertree width and its preliminaries,
discuss how it compares with other treewidth measures,
present the \emph{Intersecting Minor-Matching Theorem},
which is the main tool that makes dynamic programming
over minor matching hypertree decompositions possible,
explain how to use the Intersecting Minor-Matching Theorem,
and give examples how to solve problems using
a minor-matching hypertree decomposition.
In $\S$~3 we prove the Intersecting Minor-Matching Theorem,
and in $\S$~4 we explain how to approximate $\mu$-$tw$,
and more generally,
well-behaved hypertree width measures.

\section{Minor-matching treewidth and preliminaries}
\subsection{Notation}
All graph-theoretic notation is standard.
The \emph{induced subgraph} of $G$ by a set $S$ of vertices is denoted by $G[S]$.
A set of vertices is called \emph{independent} or \emph{stable}
if it contains no edge.
The \emph{independence number} of a graph $G$ is the
maximum size of an independent set,
and it is denoted by $\alpha(G)$.
A hypergraph $H = (V, E)$ is a pair of a set of vertices $V$ and
a set $E$ of subsets of $V$ called edges or sometimes hyperedges.
We also use $V(H)$ and $E(H)$ to denote the sets of
vertices and edges of $H$ respectively.
Independent sets and the independence number of hypergraphs
are defined the same way as in graphs,
that is sets that do not contain a hyperedge
and the maximum size of such set respectively.
If $H$ is a hypergraph, and $S \subseteq V(H)$,
we sometimes write $tr_S(H)$ instead of $tr_S(E(H))$.
For a hypergraph $H$ and $S \subseteq V(H)$,
we define the \emph{induced subhypergraph} of $H$ by $S$ to be
$H[S] = (S, \{h \in E(H) : h \subseteq S\})$.

\subsection{Tree decompositions}
\begin{define}
  [Tree decomposition~\cite{tw-history-1, tw-history-2, graph_minors_3}]
  We say that the pair $\mathcal{T} = (T, \{B_t : t \in V(T)\})$
  is a \emph{tree decomposition} of a graph $G$
  if $T$ is a tree,
  $B_t$, called the \emph{bag} of $t$,
  is a subset of $V(G)$ for each node $t$ of $T$, and
  \begin{enumerate}
  \item
    for each vertex $x$ of $G$,
    the nodes $t$ of $T$ such that $x \in B_t$, denoted $T_x$,
    induce a connected subtree of $T$, and
  \item
    every edge of $G$ is contained in some bag $B_t$, where $t \in V(T)$.
  \end{enumerate}
\end{define}
\noindent
We sometimes shorten $(T, \{B_t : t \in V(T)\})$ to
$(T, \{B_t\})$ or even $(T, B_t)$.
Instead of tree decomposition we sometimes just say decomposition.

\begin{define}
  [$\lambda$-treewidth]
  Suppose that $\lambda_G$ is a real-valued function
  on the subsets of $V(G)$
  for each graph $G$.
  Given a graph $G$ and a tree decomposition
  $\mathcal{T} = (T, \{B_t\}_{t \in V(T)})$ of $G$,
  we define
  \[
  \lambda_G(\mathcal{T}) = \max_{t \in V(T)} \lambda_G(B_t)
  \text{ } \text{ } \text{ and } \text{ } \text{ }
  \lambda\text{-}tw(G) = \min_{\mathcal{T}'} \lambda_G(\mathcal{T}'),
  \]
  to be the \emph{$\lambda$-width of $\mathcal{T}$}, and
  the \emph{$\lambda$-treewidth of $G$} respectively,
  where the minimum is taken over all tree decompositions $\mathcal{T}'$ of $G$.
\end{define}
If $G$ is clear from the context,
we normally abbreviate $\lambda_G(\mathcal{T})$ to $\lambda(\mathcal{T})$.
Technically, we ought to use $\inf$ instead of $\min$
in the definition of $\lambda\text{-}tw$,
since the set of tree decompositions is infinite,
but by restricting the definition to decompositions with distinct
bags, we obtain a finite set.
The nodes associated with the same bag induce a connected subtree,
and hence we may identify them to create a new decomposition with the
same width.
We may also require that the set of bags
forms a Sperner family without changing the width,
and it is easy to see that the size of the tree in such decomposition
is at most $n$ for an $n$-vertex graph.
Therefore we assume that every tree decomposition is either of this form,
or implicitly converted to this form,
and do not pay close attention to the size.

\begin{define}
  [Treewidth]
  The \emph{treewidth}~\cite{tw-history-1, tw-history-2, graph_minors_3}
  of a graph $G$, $tw(G)$,
  is defined as $\kappa$-$tw(G)$ for the function $\kappa_G(S) = |S| - 1$.
\end{define}

Treewidth plays a central role in the study of graph minors,
but we are mainly interested in applications to algorithms.
It is NP-complete to determine the treewidth of a graph $G$,
but for fixed $k$ there are linear-time algorithms,
with an exponential dependency on $k$,
to construct a tree decomposition of width at most $k$
or to decide that $tw(G) > k$~\cite{linear_tw}.

\begin{define}
  [Clique-width]
  The \emph{clique-width} of a graph $G$, $\cw(G)$,
  is the minimum number of labels needed to construct
  $G$ from the following operations:
  \begin{enumerate}
  \item creating of a new vertex with label $i$,
  \item disjoint union of two labelled graphs,
  \item joining by an edge every vertex with label $i$
    to every vertex of label $j$, for given $j \neq i$, and
  \item renaming every label $i$ to $j$.
  \end{enumerate}
  A sequence of these operations creating a graph $G$
  using at most $k$ labels is called \emph{$k$-expression} of $G$.
\end{define}

In contrast to treewidth, dense graphs may have low clique-width.
For example, the complete graph $K_n$ has treewidth $n-1$,
while $\cw(K_n) = 2$ for $n > 1$.
Moreover, a family of graphs with bounded treewidth
also has bounded clique-width;
more precisely, as shown in~\cite{tw_cw_comparison},
\[
\cw(G) \le 3 \cdot 2^{tw(G) - 1}.
\]
The example with the complete graph shows that no general bound exists
for the other direction,
but nevertheless, interesting special bounds do exist:
\begin{enumerate}
\item
  $tw(G) \le f(\Delta(G), \cw(G))$
  ~\cite{tw_cw_upperbound1},
\item
  $tw(G) \le g(k, \cw(G))$
  \hspace{5pt} for $K_k$-minor-free graphs $G$
  ~\cite{tw_cw_upperbound1},
\item
  $tw(G) \le 3 \cdot \cw(G)(t-1) - 1$
  \hspace{5pt} for $K_{t, t}$-free graphs $G$
  ~\cite{tw_cw_upperbound2}.
\end{enumerate}

Similarly to treewidth,
finding the clique-width of a graph is NP-complete,~\cite{clique_width_np_c}.
There is an $O(f(k)v(G)^3)$-time algorithm
finding a $(2^{k+1}-1)$-expression for a
given input graph $G$ with clique width at most $k$~\cite{approximate_cw}.

Rank-width is a closely related width measure to clique-width.

\begin{define}
  [Rank-width]
  Suppose $G$ is a graph and $A$ is a subset of $V(G)$.
  The \emph{cut-rank} of $A$, $\text{cutrk}_G(A)$,
  is the linear rank of the matrix $\{1_{ij \in E(G)} : i \in A, j \not\in A\}$.
  A \emph{ternary tree} is a tree where all non-leaves have degree $3$.
  A pair $(T, L)$ is called \emph{rank-decomposition} of $G$
  if $T$ is a ternary tree and $L$ is a bijection from the leaves of $T$
  to the vertices of $G$.
  For every edge of $e$ of $T$ the connected components of $T-e$
  induce a partition $(A, B)$ of the leaves of $T$.
  Define the \emph{width of $e$} to be
  $\text{cutrk}_G(L(A)) = \text{cutrk}_G(L(B))$,
  the \emph{width of $(T, L)$} to be the maximum width of an edge,
  and the \emph{rank-width of $G$}, $\rw(G)$,
  to be the minimum width of a rank-decomposition of $G$.
\end{define}

A key link between rank-width and clique-width has been established in
~\cite{Oum_rank_width}:
\[
\rw(G) \le \cw(G) \le 2^{\rw(G) + 1} - 1.
\]
The inequalities reveal that a family of graphs
has bounded rank-width if and only if it has bounded clique-width.

The importance of the width parameters above is best illustrated by
Courcelle's theorem.

\begin{define}
  [MSO$_1$, MSO$_2$]
  A graph can be seen as a finite model with its vertices as universe,
  and a single symmetric relation indicating adjacency.
  The \emph{monadic second-order graph logic}, MSO$_1$,
  consists of formulae for this model with quantifiers over vertices and sets of vertices.
  A graph can also be seen as a finite model with its vertices
  and edges as universe,
  and two symmetric relations indicating adjacency between vertices and edges.
  The \emph{extended monadic second-order graph logic}, MSO$_2$,
  consists of formulae for this model with quantifiers over
  vertices, edges, sets of vertices and sets of edges.
\end{define}

\begin{thm}
  [Courcelle's theorem~\cite{courcelle_thm_1, courcelle_thm_2, courcelle_thm_3}]
  For every MSO$_1$ formula $\psi_1$ and every MSO$_2$ formula $\psi_2$,
  there are functions $f_1$ and $f_2$,
  such that there are $O(f_1(cw(G)) n)$-time and $O(f_2(tw(G)) n)$-time algorithms
  to test if $G \models \phi_1$ and $G \models \phi_2$ respectively
  for every $n$-vertex graph $G$.
\end{thm}

We now turn our attention to tree decompositions of hypergraphs.

\begin{define}
  [Gaifman graph]
  The \emph{Gaifman graph} \underline{H} of a hypergraph $H$
  is a graph over the same vertex set as $H$ with two vertices
  joined by an edge whenever they appear in the same hyperedge of $H$.
\end{define}

\begin{define}
  [Tree decomposition and width of hypergraphs]
  A pair $\mathcal{T} = (T, \{B_t\}_{t \in V(T)})$
  is a tree decomposition of a hypergraph $H$
  whenever it is a tree decomposition of its Gaifman graph.
  Given a real-valued function $\lambda_H$ on the subsets of $V(H)$
  for every hypergraph $H$,
  define $\lambda_H(\mathcal{T})$ and $\lambda\text{-}tw(H)$
  the same way they are defined for graphs.
\end{define}

The first hypergraph treewidth measure for which it is possible
to construct a tree decomposition efficiently and it is not
bounded in terms of the treewidth of the Gaifman graph is
\emph{generalized hypertree width}.

\begin{define}
  \label{def:ghtw}
  [Generalized hypertree width~\cite{hypertreefirst}]
  The generalized hypertree width of a hypergraph $H$
  is defined as $\varrho\text{-}tw(H)$,
  where $\varrho_H(S)$ is the minimum size of a cover of $S$ by edges of $H$.
\end{define}

As the name suggests, there is a notion called \emph{hypertree width}.
Hypertree width, $\varrho'\text{-}tw(H)$, involves an additional technical restriction
on the tree decomposition, but as shown in~\cite{htw_invariants},
$\varrho\text{-}tw(H) \le \varrho'\text{-}tw(H) \le 3\varrho\text{-}tw(H)$,
and hence the two notions are essentially equivalent.
There is a polynomial-time algorithm
to determine if $\varrho'\text{-}tw(H) \le k$ for fixed $k$,
therefore it is possible to approximate $\varrho\text{-}tw$
within a factor of $3$ for small $k$.

\begin{define}
  \label{def:fhtw}
  [Fractional cover number]
  A \emph{fractional cover} of $ S \subseteq V(H)$ by $H$,
  is a function $\gamma : E(H) \to [0, 1]$,
  such that for each $s \in S$
  we have $\sum_{h \in E(H) : s \in h} \gamma(h) \ge 1$.
  The \emph{weight} of a fractional cover $\gamma$ is
  $\sum_{h \in E(H)} \gamma(h)$.
  The \emph{fractional cover number}, $\varrho^*_H(S)$,
  is defined as the minimum weight of a fraction cover.
\end{define}

\begin{define}
  [Fractional hypertree width~\cite{ftw}]
  We define $\varrho^*\text{-}tw$ to be \emph{fractional hypertree width}.
\end{define}
\noindent

Clearly, $\varrho^*_H(S) \le \varrho_H(S)$,
and hence $\varrho^*\text{-}tw(H) \le \varrho\text{-}tw \le \varrho'\text{-}tw$.
Given a hypergraph $H$ and a constant $k$,
it is possible to construct a tree decomposition of $H$
of $\varrho^*$-width $O(k^3)$,
or to prove that $\varrho^*\text{-}tw(H) > k$
in time $||H||^{O(k^3)}$~\cite{ftw-approx}.

A critical property of generalised hypertree width
and fractional hypertree width
is that constraint satisfaction problems can be solved efficiently
on instances of bounded width.

\begin{define}
  A \emph{constraint satisfaction problem} (CSP)
  is a triple of variables $V$, a domain $D$,
  and a set $C$ of constraints, where each constraint is
  a relation on a subset of the variables.
  The task is to assign a value from $D$ to each variable
  so that each constraint is satisfied.
\end{define}

For instance, graph colouring can be expressed as a CSP.
The domain $D$ is the set of available colours and
there is a constraint requiring each edge to be coloured properly.
We can also interpret $3$-SAT as a CSP:
the domain is $D$ is $\{0, 1\}$ and each constraint corresponds to a clause.

\begin{define}
  [Constraint hypergraph]
  The hypergraph with the variables of an instance $I$ of CSP as vertices
  and a hyperedge of each constraint is called
  \emph{constraint hypergraph} of $I$, $C(I)$.
\end{define}

\begin{thm}
  Given a tree decomposition $\mathcal{T}$ of $C(I)$,
  it is possible to solve an instance $I$ of a CSP
  in time $||I||^{\varrho(\mathcal{T}) + O(1)}$~\cite{hypertreefirst}
  and $||I||^{\varrho^*(\mathcal{T}) + O(1)}$~\cite{ftw}.
  Here $||I||$ is the sum of the number of variables,
  the number of tuples in each constraint relation,
  and the size of the domain.
\end{thm}

The rank of a hypergraph $H$, $rk(H)$,
is the maximum size of an edge of $H$.
It is easy to see that
\[
\varrho\text{-}tw(H) + 1 \le
tw(H) \le rk(H)(\varrho^*\text{-}tw(H) + 1)
\le rk(H) (\varrho\text{-}tw(H) + 1),
\]
and hence the classes of hypergraphs with restricted
rank and width in either notion are essentially equivalent.

\begin{thm}
  Let CSP($\mathcal{H}$) denote all CSP instances
  with constraint graph in a recursively enumerable class $\mathcal{H}$
  of bounded rank.
  Assuming FTW $\neq$ W[1], the following two conditions are equivalent:
  \begin{enumerate}
  \item CSP($\mathcal{H}$) is polynomially solvable,
  \item $\mathcal{H}$ has bounded treewidth.
  \end{enumerate}
\end{thm}

The condition FTW $\neq$ W[1] is a widely assumed hypothesis
in parameterised complexity.
On the contrary,
minor-matching hypertree width
cannot be bounded from below in terms of $rk(H)$ and $tw(H)$,
and therefore CSPs are computationally unfeasible.

\subsection{Minor-matching hypertree width}

To define minor-matching hypertree width,
we need some blocker machinery.
The first observation is that non-minimal edges of a hypergraph
play no role in whether a set of vertices is independent or not.

\begin{define} [Clutter]
  A hypergraph is called \emph{clutter} if no edge is a
  proper subset of another.
  Given a hypergraph $H$,
  we define $cl(H)$ to be the clutter
  containing all inclusion-wise minimal edges.
\end{define}
\noindent

Note that the complement of each independent set
is a hitting set, also called \emph{transversal}.

\begin{define} [Blocker]
  A \emph{transversal} of a clutter $H$
  is a set intersecting every edge of $H$.
  The \emph{blocker} of $H$, denoted $b(H)$,
  is the clutter consisting of all inclusion-wise
  minimal transversals of $H$.
\end{define}

One reason why $b(H)$ is more natural to work with
compared to $i(H)$ is that $b(b(H)) = H$,
while it is not generally true that $i(i(H)) = H$.

\begin{define} [Deletion and Contraction]
  Suppose $H$ is a clutter and $v \in V(H)$.
  We define $H \del v$, called \emph{$H$ with $v$ deleted},
  to be the clutter
  $(V(H) \setminus \{v\}, \{h \in E(H) : v \notin h\})$.
  We define $H \con v$, called \emph{$H$ with $v$ contracted},
  to be the clutter
  $cl(V(H) \setminus \{v\}, \{h \setminus \{v\} : h \in E(H)\})$.
\end{define}

\begin{define} [Minor]
  We say that $F$ is a \emph{minor} of $H$,
  denoted $F \minor H$,
  if $F$ can be obtained from $H$
  through a series of deletions and contractions.
\end{define}
We note that this notion of clutter minor is different from the more widely
used notion of graph minor.
By ``minor'', in this paper, we always mean clutter minor and never graph minor.
The names coincide because a graph minor renders a clutter minor in the
graphic matroid.

A critical property of minors is that deletion and contraction commute,
that is $H \del v \con u = H \con u \del v$,
$H \del v \del u = H \del u \del v$ and
$H \con v \con u = H \con u \con v$,
and therefore the order of the operations
in the definition of minor does not matter.
Even though the objects above
have a rich algebraic structure,
we only use some very basic observations:
$b(b(H)) = H$, $b(H \del v) = b(H) \con v$, $b(H \con v) = b(H) \del v$
and $F \minor H$ if and only if $b(F) \minor b(H)$.

A minor $F$ of a graph $G$ contains edges of size $0$, $1$ and $2$.
In case $F$ contains the empty set as an edge,
then $F$ contains no other edges and has no independent sets.
If $F$ contains an edge of size $1$, say $\{v\}$,
no other edge contains the vertex $v$, since $F$ is a clutter,
so these vertices play no significant role in the structure of $F$,
and we may ignore them.
The rest of $F$ is an \emph{induced} subgraph of $G$.
In particular, we see that if $F$ is a graph,
then $F$ is an induced subgraph of $G$.
We conclude that the minor and induced subgraph relations
coincide over the set of graphs.

A clutter $F$ is called \emph{matching} if it is isomorphic
to a graph matching, $F \cong kK_2$,
that is $F$ contains $k$ parallel edges of size $2$ and no other edges.

\begin{define}
  Denote by $\mu(H)$ the maximum number of edges
  in a matching which is a minor of $cl(H)$;
  and by $\mu_H(S)$ the maximum number of edges
  in a matching $F$ which is a minor of $cl(H)$,
  and every edge of $F$ intersects $S$,
  called \emph{$S$-intersecting minor-matching number}.
\end{define}

Similarly to induced subgraphs,
if $kK_2 \minor H$, then $b(kK_2) \minor b(H)$,
and hence $|i(H)| = |b(H)| \ge |b(kK_2)| = 2^k$.
In addition, if $F \minor H$,
then $|tr_S(F)| \le |tr_S(H)|$, so
$|tr_S(i(H))| \ge 2^{\mu_H(S)}$.
In this work we give a bound in the other direction, namely:
\begin{thm}
  \label{thm:Intersecting_Minor-Matching_Theorem}
  [\textbf{Intersecting Minor-Matching Theorem}]
  There exists a function $f$ such that
  for every $n$-vertex hypergraph $H$ and
  any set of vertices $S$ we have
  \[
  |tr_S(b(H))| \le n^{\mu_H(S)f(rk(H))}.
  \]
  Additionally, if $\mu_H(S)$ and $rk(H)$ are bounded,
  then $tr_S(b(H))$ can be computed in polynomial time.
\end{thm}
\noindent
The theorem is of independent interest.
Its proof is postponed to $\S$~3.

We are now ready to state the definition of minor-matching hypertree
width.

\begin{define}
  [Minor-matching hypertree width]
  We define $\mu\text{-}tw(H)$ to be the
  \emph{minor-matching hypertree width of $H$}.
\end{define}

\subsection{Solutions that can be read from the blocker}

Suppose $H$, $H_1$ and $H_2$ are hypergraphs.
We write $H = H_1 \cup H_2$ if
$V(H_1) \cup V(H_2) = V(H)$ and $E(H_1) \cup E(H_2) = E(H)$.
If $H = H_1 \cup H_2$,
we say that $S$ is an $(H_1, H_2)$-separator if
$V(H_1) \cap V(H_2) \subseteq S$.

Suppose we want to compute a parameter $f_H$.
If it is possible to generalise $f_H$ to a function
$f_H^S$ over $tr_S(i(H))$,
say $f_H = f_H^\emptyset(\emptyset)$,
for any $S \subseteq V(H)$,
and it is possible to compute
$f_H^S$ from $f_{H_1}^S$ and $f_{H_2}^S$, where $S$ is an $(H_1, H_2)$-separator,
then, roughly speaking, we can compute $f_H$
given a tree decomposition $\mathcal{T}$ with $\mu(\mathcal{T}) \le k$.
The next lemma formalises the discussion above.

\begin{define}
  \label{def:can_be_read_from_blocker}
  We say that a function $f_H^S : tr_S(i(H)) \to \mathbb{Q}$,
  or more generally, $f_H^S : tr_S(i(H))^p \to \mathbb{Q}$,
  \emph{can be read from the blocker} if it has
  the following properties:
  \begin{enumerate}
  \item
    it is possible to compute $f_H^{V(H)}$ from $i(H)$ efficiently,
  \item
    it is possible to compute $f_H^{S}$ from $f_H^{S'}$
    efficiently if $S \subseteq S'$,
  \item
    it is possible to compute $f_{H+v}^{S \cup \{v\}}$ from  $f_H^S$,
    where $H+v$ is a hypergraph obtained from $H$
    by adding a new isolated vertex, $v$,
  \item
    it is possible to compute $f_H^S$
    from  $tr_S(i(H))$, $f_{H_1}^S$ and $f_{H_2}^S$ efficiently
    if $V(H_1) \cap V(H_2) \subseteq S$ and $H = H_1 \cup H_2$,
  \end{enumerate}
  for every hypergraph $H$ and set $S \subseteq V(H)$.
  Here efficiently means in polynomial time
  with respect to the size of the input, including $tr_S(i(H))$,
  and the size of the output
  where all functions are represented in table form.
\end{define}

\begin{lemma}
  Suppose $r$ and $k$ are fixed constants.
  Let $H$ be a hypergraph of rank at most $r$,
  $\mathcal{T} = (T, \{B_t\})$ be a tree decomposition of $H$
  with $\mu(\mathcal{T}) \le k$,
  and $f_H^S$ be a function that can be read from the blocker.
  It is possible to compute $f_H^{B_t}$ from $(H, \mathcal{T})$
  in polynomial time w.r.t. $||H||$ for every bag $B_t$,
  $t \in T$.
  In particular, it is possible to compute $f_H^\emptyset$ in polynomial time.
\end{lemma}
\begin{proof}
  Given a bag $B_r$, pick $r$ as the root of $T$
  and orient every edge from the root towards the leaves.
  Let $T_t$ be the subtree of $T$ composed of descendants of $t$,
  and let $H_t$ be the subhypergraph of $H$ induced by the
  vertices contained in bags of $T_t$.
  We compute $f_{H_t}^{B_t}$ recursively, from the leaves towards the root.
  We can compute $f_{H_l}^{B_l}$ for each leaf $l$ of $T$ using
  property $1.$ of $f$.
  Given a node $t$ with direct descendants $t_1 \ldots t_m$,
  first compute $f_{H_{t_i}}^{B_t \cap B_{t_i}}$
  using property $2.$,
  then add each vertex in $B_t \setminus B_{t_i}$ as an isolated vertex
  to compute $f_{H_{t_i} + (B_t \setminus B_{t_i})}^{B_t}$
  using property $3.$, and finally merge sequentially $f_{B_t}^{B_t}$ with
  $f_{H_{t_i} + (B_t \setminus B_{t_i})}^{B_t}$
  to construct $f_{H_t}^{B_t}$.
  Theorem~\ref{thm:Intersecting_Minor-Matching_Theorem}
  guarantees that $tr_{B_t}(i(H))$ is polynomially sized
  and that it can be computed in polynomial time,
  and hence the input and the output of each step is polynomially sized.
\end{proof}

\subsection{Examples of functions that can be read from the blocker}
\subsubsection{Maximum Weighted Independent Set}
Our first example is \textbf{Maximum Weighted Independent Set}.
On instance hypergraph $H$ and a weight function
$w : V(H) \to \mathbb{Q}$,
the problem is to find an independent set $J$ of $H$
with maximum value $w(J) = \sum_{j \in J} w(j)$.
Delete all vertices $v$ with $w(v) < 0$,
and now we may restrict our attention to maximal independent sets $J \in i(H)$.
Define $f_H^S(A) := \max\{w(J) : J \in i(H), A = S \cap J\}$.
It suffices to show that $f_H^S$ can be read from the blocker
in order to prove that the problem can be solved in polynomial
time using a tree decomposition of bounded $\mu$-width.
Properties 1.-3. are easy to verify.
Suppose $H = H_1 \cup H_2$ and $V(H_1) \cap V(H_2) \subseteq S$
and that $f_H^S(A)$ is witnessed by $J \in i(H)$.
Since $J$ is stable in $H_1$ and $H_2$,
we may find $J_i \in i(H_i)$ such that $J \cap V(H_i) \subseteq J_i$.
In particular, $A \subseteq J_i \cap S$.
We have
\[
w(J) = w(J_1) + w(J_2) - w(J_1 \cap S) - w(J_2 \cap S) + w(A).
\]
It follows that $w(J_i) = f_{H_i}^S(J_i \cap S)$,
and hence
\begin{align*}
f_H^S(A) = \max\{
  & f_{H_1}^S(A_1) + f_{H_2}^S(A_2) - w(A_1) - w(A_2) + w(A) : \\
  & A_1 \in tr_S(i(H_1)), A_2 \in tr_S(i(H_2)), A = A_1 \cap A_2
\}.
\end{align*}

\subsubsection{k-Colouring}
A more interesting example is \textbf{$k$-Colouring}.
A colouring with $k$ colours is a function $c : V(H) \to [k]$.
A colouring $c$ is called \emph{proper} if no edge is \emph{monochromatic},
that is all of its vertices are coloured the same.
The chromatic number of $H$, $\chi(H)$, is the minimum number of colours
in a proper colouring.

We need to phrase the problem is terms of maximal independent sets
in order to tackle it.
Suppose $c: V(H) \to [k]$ is a proper colouring.
Observe that $c^{-1}(a)$ is an independent set for each colour $a$.
Therefore $\chi(H)$ is the minimum size of a cover of $H$ with independent sets.
We may assume that each independent set in the cover is maximal.
Given $i(H)$, we can easily see if $\chi(H) \le k$:
simply check if any $k$-tuple $(I_1, I_2, \ldots, I_k)$ of independent sets
covers $H$.

Consider the function
\[
f_H^S(A_1, A_2, \ldots, A_k)
=
\begin{cases*}
  1 & if $(A_1, A_2, \ldots, A_k)$ covers $S$ and \\
    & there exists a cover $(I_1, I_2, \ldots, I_k)$ of $H$, \\
    & s.t. $I_i \in i(H)$ and $A_i \subseteq I_i$ for $1 \le i \le k$; \\
  0 & otherwise.
\end{cases*}
\]

It is easy to see that $f_H^S$ satisfies properties 1.-3. of
Definition~\ref{def:can_be_read_from_blocker}.

\begin{lemma}
  Suppose $H = H_1 \cup H_2$, $V(H_1) \cap V(H_2) \subseteq S$,
  and $\{A_i \in tr_S(i(H))\}_{i \in [k]}$ covers $S$.
  The following two conditions are equivalent:
  \begin{enumerate}
  \item
    $f_H^S(A_1, A_2, \ldots, A_k) = 1$,
  \item
    we can find $A_i^j \in tr_S(i(H_j))$,
    $i \in [k]$, $j \in [2]$,
    such that $A_i \subseteq A_i^1 \cap A_i^2$,
    and $f_{H_j}^S(A_1^j, A_2^j, \ldots, A_k^j) = 1$
    for $j \in [2]$.
  \end{enumerate}
\end{lemma}

\begin{proof}
  First suppose that $f_H^S(A_1, A_2, \ldots, A_k) = 1$ and
  that it is witnessed by $(I_1, I_2, \ldots, I_k)$.
  We see that $J_i^j = V(H_j) \cap I_i$ is stable in $H_j$,
  hence we may find $I_i^j \in i(H_j)$ such that $J_i^j \subseteq I_i^j$.
  Observe that $(I_1^j, I_2^j, \ldots, I_k^j)$ covers $H_j$.
  Therefore $f_{H_j}^S(I_1^j \cap S, I_2^j \cap S, \ldots, I_k^j\cap S) = 1$.

  Conversely, suppose that the second condition holds
  and that it is witnessed by $(I_1^j, I_2^j, \ldots, I_k^j)$.
  We may write $J_i = (I_i^1 \cup I_i^2) \setminus S \cup A_i$
  to cover $H$ with $k$ (not necessarily maximal) independent sets.
  Extend each $J_i$ to a maximal independent set $I_i \in i(H)$
  to complete the proof.
\end{proof}

The lemma proves that $f_H^S$ satisfies condition 4. of
Definition~\ref{def:can_be_read_from_blocker}.
Clearly $f_H^\emptyset(\emptyset^k) = 1$ if and only if $\chi(H) \le k$.

\subsubsection{Hypergraph Homomorphism}
Suppose $H$ and $F$ are hypergraphs.
A (not-necessarily injective) mapping $f : V(H) \to V(F)$
is called \emph{homomorphism} if the image of each edge of $H$ is an edge of $F$.
The existence of a homomorphism between $H$ and $F$ is denoted by $H \to F$.
The framework of homomorphisms has a rich expressive power to describe
graph-theoretic properties.
For instance, $\chi(H) \le k$ if and only if $H \to K_k$,
where $K_k$ is the complete $k$-vertex graph.

As usual, we start by formulating the problem is terms of maximal
independent sets.
A hypergraph is called \emph{uniform} if all edges have the same size.

\begin{lemma}
  Suppose $H$ and $F$ are $r$-uniform hypergraphs.
  A mapping $f : V(H) \to V(F)$ is a homomorphism
  if and only if $f^{-1}(I) \in \hat{i}(H)$ for every $I \in i(F)$.
\end{lemma}
\begin{proof}
  If $f$ is a homomorphism,
  then clearly $f^{-1}(I)$ is stable for each $I \in i(F)$,
  as otherwise the image of any edge contained in $f^{-1}(I)$ is not an edge.

  Suppose $f : V(H) \to V(F)$ is a function such that
  $f^{-1}(I) \in \hat{i}(H)$ for every $I \in i(F)$.
  Let $h$ be an edge of $H$ and let $f(h)$ be its image.
  Since $H$ and $F$ are $r$-uniform, $f(h)$ is either an edge of $F$ or stable.
  If $f(h)$ is stable, let $I \in i(F)$ be a maximal independent set of $F$
  containing $f(h)$.
  But now we see that $h \subseteq f^{-1}(I)$, which is a contradiction,
  and hence $f(h)$ is an edge for every $h \in E(H)$.
\end{proof}

\begin{lemma}
  Suppose $H$ and $F$ are $r$-uniform hypergraphs.
  The following two statements are equivalent:
  \begin{enumerate}
  \item $H \to F$,
  \item there exists a function $g : i(F) \to i(H)$
    such that
    \[
    \left\{ A_x = \bigcap_{ I \in i(F), x \in I} g(I): x \in V(F) \right\}
    \]
    covers $V(H)$.
  \end{enumerate}
\end{lemma}
\begin{proof}
  First suppose $f : H \to F$ is a homomorphism.
  Each $f^{-1}(I)$ is stable in $H$ for $I \in i(F)$,
  so we may set $g(I) = J$, where $J \in i(H)$
  is an arbitrary maximal independent set containing $f^{-1}(I)$.
  Let $v$ be an arbitrary vertex of $H$.
  We have $v \in A_{f(v)}$,
  and therefore the sets $A_x$ cover $V(H)$.

  Now suppose $g$ is a function satisfying condition 2. in the statement.
  For each vertex $v$ of $H$ pick some arbitrary vertex $x$ of $F$
  such that $v \in A_x$ and set $f(v) = x$.
  Such a vertex exists, because $\{A_x\}_{x \in V(F)}$ covers $V(H)$.
  Let $I \in i(F)$ be arbitrary.
  We see that $f^{-1}(I) \subseteq g(I)$ is stable in $H$,
  and therefore $f$ is a homomorphism.
\end{proof}

Now suppose $F$ is a fixed $r$-uniform hypergraph
and on instance $r$-uniform hypergraph $H$ we have to decide if $H \to F$.
Suppose $i(F) = \{I_1, \ldots, I_k\}$.
Let $f_H^S : tr_S(i(H))^k \to \{0, 1\}$ be defined as follows:
\[
f_H^S(A_1, \ldots, A_k) = 1 \text{ iff}
\]
\begin{enumerate}
\item
there exists a $k$-tuple $(J_1, \ldots, J_k) \in i(H)^k$
such that $\left\{ \cap_{i : x \in I_i} J_i\right\}_{x \in V(F)}$ covers $H$;
\item $\left\{ \cap_{i : x \in I_i} A_i\right\}_{x \in V(F)}$ covers $S$; and
\item $A_i \subseteq J_i$.
\end{enumerate}
Clearly $f_H^\emptyset(.) = 1$ iff $H \to F$
and $f_H^S$ satisfies properties 1. - 3. of
Definition~\ref{def:can_be_read_from_blocker}.
Property 4. follows from the following lemma.

\begin{lemma}
  Suppose $H = H_1 \cup H_2$, $V(H_1) \cap V(H_2) \subseteq S$,
  $A_i \in tr_S(i(H))$ for $1 \le i \le k$
  and $\left\{ \cap_{i : x \in I_i} A_i\right\}_{x \in V(F)}$ covers $S$.
  The following two conditions are equivalent:
  \begin{enumerate}
  \item
    $f_H^S(A_1, \ldots, A_k) = 1$,
  \item
    we can find $A_i^j \in tr_S(i(H_j))$,
    $i \in [k]$, $j \in [2]$,
    such that $A_i \subseteq A_i^1 \cap A_i^2$,
    and $f_{H_j}^S(A_1^j,  \ldots, A_k^j) = 1$
    for $j \in [2]$.
  \end{enumerate}
\end{lemma}

\begin{proof}
  First suppose $(J_1, \ldots, J_k)$ witnesses $f_H^S(A_1, \ldots, A_k) = 1$.
  We can find $(J_1^j, \ldots, J_k^j) \in i(H_j)^k$ such that
  $J_i \cap V(H_j) \subseteq J_i^j$.
  It can be checked that
  $f_{H_j}^S(J_1^j \cap S,  \ldots, J_k^j \cap S) = 1$.

  Now suppose the second condition in the statement holds and that it is
  witnessed by $(J_1^j, \ldots, J_k^j) \in i(H_j)^k$
  and $(A_1^j, \ldots, A_k^j) \in tr_S(i(H_j))^k$.
  Find a $k$-tuple $(J_1, \ldots J_k) \in i(H)^k$
  such that $(J_i^1 \cup J_i^2)\setminus S \cup A_i \subseteq J_i$.
  It is easy to check that $(J_1, \ldots J_k)$ witnesses
  $f_H^S(A_1, \ldots, A_k) = 1$.
\end{proof}

\subsection{Comparison with other tree decompositions}
As discussed before, treewidth, hypertree width and fractional hypertree width
are functionally bounded to each other for hypergraphs of bounded rank,
while minor-matching hypertree width cannot be bounded from below by treewidth
even for graphs.
For instance, consider the complete graph, $K_n$.
Since $K_n$ is $2K_2$-free, we have $\mu\text{-}tw(K_{n}) = 1$,
while $tw(K_n) = n-1$.

As for clique-width (or equivalently rank-width),
it cannot be bounded functionally by $\mu\text{-}tw$.
For instance,
consider the graph $G$ over $Z_{4k+1}$ where $i \sim j$
if $i > k+j \mod 4k+1$.
We have that $G$ is $2K_2$-free, hence $\mu\text{-}tw(G) = 1$,
while $cr(G) \ge k$.
In the other direction,
let $G_0 = K_1$,
let $G_{2n+1}$ be a complete union of two copies of $G_{2n}$
and let $G_{2n+2}$ be a disjoint union of two copies of $G_{2n+1}$.
Clearly $rkw(G_n) = 1$,
while using the tools developed in \S~4 we can see that
$\mu\text{-}tw(G_{2n+1}) = n$.

Nevertheless, minor-matching hypertree width permits many more graphs
of bounded width and additionally it applies to hypergraphs as well as graphs.
Indeed, every co-bipartite graph (graph $G$ with $\chi(\overline{G}) \le 2$)
is $3K_2$-free, and there are $e^{\Theta(n^2)}$ $n$-vertex co-bipartite graphs,
while the number of graphs with bounded clique-width by a fixed constant
is $e^{O(n \ln n)}$.

We can give a more precise estimate on the size of the class of
$n$-vertex graphs $G$ with $\mu\text{-}tw(G) \le k$, call it $M_{n, k}$.
Note that $M_{n, k}$ contains all graphs $G$ with $\chi(\overline G) \le k$,
since $\mu(G) \le \alpha(G) \le \chi(\overline G) \le k$;
does not contain all bipartite graphs, not all split graphs,
and not all graphs with $\chi(\overline G) \le k + 1$ (for large $n$);
and therefore the colouring number of $M_{n, k}$ is exactly $k$ (for large $n$) and
$|M_{n, k}| = 2^{(1 - \frac1k + o(1)){n \choose 2}}$
by the colouring number theorem~\cite{alekseev_size, bollobas97, bollobas09, bollobas11}.
Similar argument holds for graphs $G$ with $\alpha\text{-}tw(G) \le k$.

\begin{lemma}
  The number of $n$-vertex graphs $G$ with $\mu\text{-}tw(G) \le k$
  ($\alpha\text{-}tw(G) \le k$) is $2^{(1 - \frac1k + o(1)){n \choose 2}}$.
\end{lemma}

\section{Proof of the Intersecting Minor-Matching Theorem}
This section is dedicated to proving the Intersecting Minor-Matching Theorem,
Theorem~\ref{thm:Intersecting_Minor-Matching_Theorem}.
We use the machinery developed in~\cite{yolov2016blocker}.
We need preliminary definitions first.

\begin{define} [Join and Meet]
  Suppose $H$ and $F$ are clutters.
  Define
  \begin{align*}
    H \vee F &:= cl(V(H) \cup V(F), E(H) \cup E(F)), \\
    H \wedge F &:=
    cl(V(H) \cup V(F), \{h \cup f : h \in E(H), f \in E(F)\}).
  \end{align*}
\end{define}

Join and meet are dual: $b(H \vee F) = b(H) \wedge b(F)$ and
$b(H \wedge F) = b(H) \vee b(F)$;
and they commute with deletion and contraction:
$(H \vee F) \del v = (H \del v) \vee (F \del v)$.
For a pair of disjoint sets of vertices $S$ and $T$ of $H$,
we denote the minor of $H$ obtained by deleting $S$ and contracting $T$
by $H[S; T]$.

\begin{define}
Given a clutter $H$, and edge $h \in E(H)$ and $v, u \neq v \in h$,
define $H \circ (h, u, v) = H[u; h \setminus \{u\}] \lor H[v; h \setminus \{v\}]$.
We call a sequence $\Pi = \{(h_i, u_i, v_i)\}_{i=1}^n$
a \emph{quasimatching} if
\begin{enumerate}
  \item
    $\{v_i, u_i\} \subseteq h_i$ and $v_i \neq u_i$ for each $i \in [n]$,
  \item $h_{i+1}$ is an edge of
    $(H \circ (h_1, v_1, u_1)) \circ \ldots \circ (h_i, v_i, u_i)$.
\end{enumerate}
A quasimatching $\Pi$ is called $S$-intersecting if each pair $\{v_i, u_i\}$
intersects $S \subseteq V(H)$.
\end{define}

\begin{lemma}
  \label{thm:transversal_types}
  Suppose we are given a clutter $H$, a minimal transversal $T$ of $H$,
  a set $S \subseteq V(H)$ and a vertex $x$ of $H$.
  Either
  \begin{enumerate}
  \item
    (i) there exists a minimal transversal $T' \subseteq T \setminus \{x\}$
    of $H \del x$ such that $T' \cap S = (T \setminus \{x\}) \cap S$;
    or
  \item
    (ii) there exists a vertex $z \in S \setminus \{x\}$
    and an edge $h \in E(H)$
    such that $x \in h$,
    $h \cap T = \{z\}$
    and $T \setminus \{z\} \in b(H \circ (h, x, z))$.
  \end{enumerate}
\end{lemma}
\begin{proof}
  \underline{Case 1:}
  For each $z \in (T \setminus \{x\}) \cap S$
  there exists $h \in E(H \del x)$ such that $h \cap T = \{z\}$.

  Let $T^\prime \subseteq T\setminus \{x\}$ and $T^\prime \in b(H \del x)$.
  Clearly $T^\prime \cap S \subseteq (T \setminus \{x\}) \cap S$.
  To prove the opposite direction,
  suppose $t \in (T \setminus \{x\}) \cap S$.
  We can find $h \in E(H \del x)$ such that $h \cap T = \{t\}$.
  But since
  $T'$ is a transversal of $H \del x$ -- hence $h \cap T'$ is non-empty --
  and $T' \subseteq T$,
  it follows that $h \cap T' = \{t\}$,
  and therefore $t \in T' \cap S$.

  \underline{Case 2:}
  There exists $z \in (T \setminus \{x\}) \cap S$
  such that for each $h \in E(H)$
  if $h \cap T = \{z\}$, then $x \in h$.

  Since $z \in T$ and $T$ is minimal,
  there is $h \in E(H)$ such that $h \cap T = \{z\}$.
  It is easy to see that $T \setminus \{z\} \in b(H[z, h \setminus \{z\}])$.
  Let $f \in H[x, h \setminus \{x\}]$, and let $\widehat{f} \in H$
  be such that $f \subseteq \widehat{f} \subseteq f \cup (h \setminus \{x\})$.
  Since $x \not\in \widehat{f}$
  we have $\widehat{f} \cap T \neq \{z\}$.
  Let $y \in \widehat{f} \cap (T \setminus \{z\})$,
  and hence $y \notin h$, so $y \in f$.
  But then $(T \setminus \{z\}) \cap f$ is non-empty,
  and $f$ is arbitrary,
  so we conclude that $T \setminus \{z\}$
  is (possibly non-minimal) transversal of $H[x, h \setminus \{x\}]$.
  The last statement in conjunction with $T \setminus \{z\} \in b(H[z, h \setminus \{z\}])$
  implies that $T \setminus \{z\} \in b(H \circ (h, x, z))$.
\end{proof}

\begin{lemma}
  \label{thm:quasimatchings}
  The size of $tr_S(b(H))$ is at most the number of $S$-intersecting
  quasimatchings.
\end{lemma}
\begin{proof}
  We argue by induction.
  The claim trivially holds for clutters with one vertex.
  Suppose the statement holds for clutters with less than $n$ vertices.

  If $V(H) \not\subseteq S$, pick some $x \in V(H) \setminus S$,
  otherwise pick arbitrary $x \in V(H)$.
  Let $\mathcal{T}_1 \subseteq E(b(H))$ be the set of minimal transversals
  of the first type described in Lemma~\ref{thm:transversal_types},
  and let $\mathcal{T}_2 = E(b(H)) \setminus \mathcal{T}_1$.

  If $x \notin S$, then clearly
  $tr_S(\mathcal{T}_1) \subseteq tr_S(b(H \del x))$.
  On the other hand, let $A \in tr_S(b(H \del x))$ and let $T' \in b(H \del x)$,
  such that $T' = A \cap S$.
  In particular, we can find $T \in \{T', T' \cup \{x\}\} \cap b(H)$.
  It follows that $T \in \mathcal{T}_1$,
  $A \in tr_S(\mathcal{T}_1)$ and finally
  $tr_S(b(H \del x)) \subseteq tr_S(\mathcal{T}_1)$.
  We conclude that $tr_S(\mathcal{T}_1) \equiv tr_S(b(H \del x))$,
  and therefore $|tr_S(\mathcal{T}_1)| = |tr_S(b(H \del x))|$.

  If $V(H) \subseteq S$,
  then the mapping $T \mapsto T \setminus \{x\}$
  is a bijection between $\mathcal{T}_1$ and $b(H \del x)$,
  and hence
  \[
  |tr_S(\mathcal{T}_1)| = |\mathcal{T}_1|
  = |b(H \del x)| = |tr_S(b(H \del x))|.
  \]
  We see that $|tr_S(\mathcal{T}_1)| = |tr_S(b(H \del x))|$
  in either case.

  The transversals of $\mathcal{T}_2$ fall in the second case of
  Lemma~\ref{thm:transversal_types}.
  Each $A \in tr_S(\mathcal{T}_2)$ can be found
  in $tr_S(b(H \circ (h, x, z))) + \{z\}$,
  where ``$+$'' is an abuse of notation meaning that the vertex $z$
  is added to each set in the trace,
  for an appropriate choice of a triple $(h, x, z)$.

  Finally we have
  \[
  |tr_S(b(H))| \le |tr_S(b(H \del x))|
  +   \sum_{z \in S \setminus \{x\}}
  \sum_{h \in E(H) : \{x, z\} \subseteq h}
  |tr_S(b(H \circ (h, x, z)))|.
  \]
  Let $qm_S(H)$ denote the set of $S$-intersecting quasimatchings of $H$.
  By induction we have $|tr_S(b(H \del x))| \le |qm_S(H \del x)|$,
  and $|tr_S(b(H \circ (h, x, z)))| \le |qm_S(H \circ (h, x, z))|$.
  In addition, for every $\Pi' \in qm_S(H \circ (h, x, z))$
  we may construct $\Pi = \{(h, x, z)\} \cup \Pi' \in qm_S(H)$.
  Any two quasimatchings constructed this way are distinct,
  as they either disagree on the triple $(h, x, z)$
  (or the lack of such triple containing $x$)
  or are constructed from distinct members $qm_S(H \circ (h, x, z))$
  (or $qm_S(H \del x)$), and hence are distinct.
  We conclude that $|tr_S(b(H))| \le |qm_S(H)|$.
\end{proof}

\begin{lemma}
  \label{thm:matching_thm}
  Suppose $H$ is a rank $r$ clutter and $\Pi = \{(h_i, u_i, v_i)\}_{i=1}^\ell$
  is a quasimatching of size $\ell$.
  There exists a subsequence $j_1, \ldots, j_k$ of size
  $k \ge \ell 2^{-(r-2)}/(2r-3)$, such that
  \[
  kK_2 \cong \{\{u_{j_i}, v_{j_i}\}\}_{i=1}^k \minor H.
  \]
\end{lemma}
\begin{proof}
  It suffices to verify that quasimatchings
  are semi-matchings as defined in \cite{yolov2016blocker}, Definition 3.1,
  and the lemma follows from
  Theorem 3.4 of~\cite{yolov2016blocker}.
\end{proof}

\begin{proof}
  [Proof of the Intersecting Minor-Matching Theorem,
  Theorem~\ref{thm:Intersecting_Minor-Matching_Theorem}]

  Suppose $m = \mu_H(S)$ and $r = rk(H)$.
  The number of choices for a triple $(h, v, u)$ is bounded by $n^{r}$.
  The number of sequences of such triples of length at most $\ell$ is
  $\sum_{k=1}^\ell n^{rk} \le n^{r(\ell + 1)}$.
  If we set $\ell$ to be the length of the longest $S$-intersecting
  quasimatching, we see by Lemma~\ref{thm:matching_thm} that
  $\ell \le m 2^{r-2}(2 r - 3)$.
  By Lemma~\ref{thm:quasimatchings} we conclude that
  \[
  |tr_S(b(H))| \le n^{m2^{r-2}(2 r^2 - 3r) + r}.
  \]

  To make the theorem algorithmic,
  one may consider a branching algorithm,
  exploring all cases of Lemma~\ref{thm:transversal_types}.
  The size of the recursion tree is bounded by the number
  of quasimatchings, hence the total running time is polynomial.
\end{proof}

\section{Approximation of $\mu\text{-}tw$ for graphs
  and $\alpha\text{-}tw$ for hypergraphs}
\label{sec:approx_mu}
This section is organised as follows.
We start by arguing that finding decompositions
of low $\alpha$ and $\mu$ width for graphs is essentially equivalent.
In \S~\ref{sec:red_alpha_to_mu} we explain how to reduce a graph $G$
to a graph $M(G)$ such that $\alpha\text{-}tw(G) = \mu\text{-}tw(M(G))$.
Moreover, for every decomposition $\mathcal{T}$ of $M(G)$
we can find a decomposition $M^{-1}(\mathcal{T})$ of $G$ such that
$\alpha(M^{-1}(\mathcal{T})) \le \mu(\mathcal{T})$.
The main result of \S~\ref{sec:red_mu_to_alpha}
states that a reduction from $\mu\text{-}tw$ to $\alpha\text{-}tw$
exists with the same properties.

Recall that $\underline{H}$ is the Gaifman graph of $H$.
\begin{define}
  We say that a hypergraph width measure $\lambda_H$ is \emph{well-behaved} if
  \begin{enumerate}
  \item
    $\lambda_H(\{x\}) \ge 1$
    for each vertex $x \in V(H)$,
  \item
    $\lambda_H(S \cup T) \le \lambda_H(S) + \lambda_H(T)$,
  \item
    $\lambda_H(S \cup T)
    = \lambda_H(S) + \lambda_H(T)$
    for disjoint sets $S$ and $T$ such that
    $E_{\underline{H}}(S, T) = \emptyset$,
  \item
    $\lambda_{F}(S) \le \lambda_{H}(T)$
    for each pair of hypergraphs $F \supseteq H$
    and sets $S \subseteq T$, and
  \item
    $\lambda_H(S) \le k$ can be decided in $||H||^{O(k)}$ time.
  \end{enumerate}
\end{define}

Note that the independence number $\alpha_H$,
the edge covering number $\varrho_H$ (Definition~\ref{def:ghtw}),
and the fractional edge covering number $\varrho^*_H$
(Definition~\ref{def:fhtw}),
are all well-behaved,
while $\mu_H$ is not.
The following theorem is proved in \S~\ref{sec:approx_well_behaved_msr}.
\begin{thm}
  \label{thm:approx_well_behaved}
  Suppose $\lambda_H$ is a well-behaved width measure for hypergraphs.
  There exists an algorithm with running time $||H||^{O(k^3)}$ that
  given a hypergraph $H$ and an integer $k$ as input, either
  \begin{itemize}
  \item
    finds a tree decomposition $\mathcal{T}$ of $H$
    with $\lambda(\mathcal{T}) = O(k^3)$, or
  \item
    correctly concludes that $\lambda\text{-}tw(H) > k$.
  \end{itemize}
\end{thm}
\noindent
Note the lack of restrictions on the rank.

The overall approximation for $\mu\text{-}tw$ is a follows:
staring from a graph $G$ and a number $k$,
build a graph $L = L(G)$ as in \S~\ref{sec:red_mu_to_alpha}.
We have $\mu\text{-}tw(G) = \alpha\text{-}tw(L)$,
so the algorithm in Theorem~\ref{thm:approx_well_behaved}
either finds a decomposition $\mathcal{T}$ of $L$
such that $\alpha(\mathcal{T}) = O(k^3)$, or declares
$\alpha\text{-}tw(L) > k$.
In former case we get a decomposition $L^{-1}(\mathcal{T})$
of $G$ such that $\mu(L^{-1}(\mathcal{T})) = O(k^3)$,
in the latter we deduce that $\mu\text{-}tw(G) > k$.

\subsection{Reducing $\alpha\text{-}tw$ to $\mu\text{-}tw$}
\label{sec:red_alpha_to_mu}
Given a graph $G$ let $M = M(G)$ be obtained from $G$
by adding a new vertex $v'$ and an edge $vv'$ for each $v \in V(G)$.
The key observation here is that given a set $S \subseteq V(G) \subseteq V(M)$,
we have $\alpha_G(S) = \mu_{M}(S)$.

\begin{lemma}
  We have $\alpha\text{-}tw(G) = \mu\text{-}tw(M)$.
\end{lemma}
\begin{proof}
  If $\mu\text{-}tw(M) = 0$, then $M$ contains no vertices,
  but in this case $\alpha\text{-}tw(G) = \mu\text{-}tw(M) = 0$.
  Suppose $\mathcal{T} = (T, \{B_t\})$ is a tree decomposition of $G$
  of $\alpha$-width at least $1$.
  Extend $\mathcal{T}$
  by adding new nodes with bags $vv'$ for each vertex $v$
  and attached to an arbitrary node containing $v$,
  or an arbitrary node if no such node exists.
  It is easy to see that
  $\mathcal{M}(\mathcal{T})$ is a tree decomposition of $M$.
  The $\mu$-width of each new bag is $1$,
  hence a bag of maximum width can be found among the nodes of $\mathcal{T}$.
  As noted before,
  the $\alpha$-width of a node of $\mathcal{T}$
  equals its $\mu$-width in $M$,
  hence $\mu(\mathcal{M}(\mathcal{T})) = \alpha(\mathcal{T})$,
  and therefore, since $\mathcal{T}$ was arbitrary,
  we have $\alpha\text{-}tw(G) \ge \mu\text{-}tw(M)$.

  Conversely, let $\widehat{\mathcal{T}}$ be a tree decomposition of $M$.
  For each vertex $v$ of $G$
  delete $v'$ from each bag of $\widehat{\mathcal{T}}$
  and add a new node with bag $vv'$
  attached to an arbitrary node of $\widehat{\mathcal{T}}$ containing $v$,
  or an arbitrary node if no such node exists.
  We see that the modified tree decomposition is of the form
  $\mathcal{M}(\mathcal{T})$ for a tree decomposition $\mathcal{T}$ of $G$.
  It is easy to see that
  $\mu(\mathcal{M}(\mathcal{T})) \le \mu(\widehat{\mathcal{T}})$,
  and in the previous paragraph we saw that
  $\alpha(\mathcal{T}) =\mu(\mathcal{M}(\mathcal{T}))$,
  hence $\mu\text{-}tw(M) \ge \alpha\text{-}tw(G)$.
\end{proof}

Combining $\mu\text{-}tw(G) \le \alpha\text{-}tw(G)$
with the lemma above we get
\[
\mu \text{-} tw(G) \le \alpha \text{-} tw(G) =
\mu \text{-} tw(M(G)) \le \alpha \text{-} tw(M(G)).
\]
For completeness, we note
\[
\mu(G) \le \alpha(G) = \mu(M(G)) \le \alpha(M(G)).
\]

\subsection{Reducing $\mu\text{-}tw$ to $\alpha\text{-}tw$}
\label{sec:red_mu_to_alpha}
Recall that the \emph{line graph} of $G$, denoted $L(G)$,
has the edges of $G$ as vertices,
and two vertices of $L(G)$ are adjacent
whenever the corresponding edges of $G$ intersect.

Let $L^k(G)$ be the graph with vertices $E(G)$ and edges
\[
\left\{ef \in {E(G) \choose 2} : 1 \le d_{L(G)}(e, f) \le k\right\}.
\]
Here $d_G(u, v)$ is used to denote the \emph{distance} between
$u$ and $v$ in $G$,
that is the length of the shortest $u$--$v$ path in $G$
measured in number of edges.
If there is no $u$--$v$ path, we set $d_G(u, v) = \infty$.
Note that $L(G) = L^1(G)$.
We focus on $L = L^2(G)$.
Two vertices, $e$ and $f$, of $L$ are adjacent
if $G[e \cup f]$ is connected, or equivalently $G[e \cup f] \not\cong 2K_2$.
Let $\mathcal{L} : 2^{V(G)} \to 2^{V(L)}$
map sets $S \subseteq V(G)$ to the edges of $G$ incident with $S$.
We see that $\mu_G(S) = \alpha_L(\mathcal{L}(S))$.

\begin{thm}
  We have $\mu\text{-}tw(G) = \alpha\text{-}tw(L)$.
\end{thm}
\begin{proof}
Suppose $\mathcal{T} = (T, B_t)$ is a tree decomposition of $G$.
Consider $\mathcal{L}(\mathcal{T}) = (T, \{\mathcal{L}(B_t)\}_{t \in V(T)})$.
If $ef$ is an edge of $L$,
then either $e$ and $f$ share a vertex,
or there is an edge $g$ intersecting both.
Let $h$ be $e$ in the prior case and $g$ in the latter.
The edge $h$ must be contained in some bag $B_{t_0}$,
and note that $\mathcal{L}(B_{t_0})$ contains $ef$.
Denote by $T_e$ the nodes of $\mathcal{L}(\mathcal{T})$
with bags containing $e = uv \in V(L)$.
Observe that $T_e = T_v \cup T_u$;
$T_v$ and $T_u$ are trees that touch, hence $T_e$ is connected.
We conclude that $\mathcal{L}(\mathcal{T})$ is a tree decomposition of $L$.
For each bag $B_t$ we have $\mu_G(B_t) = \alpha_L(\mathcal{L}(B_t))$,
so $\mu(\mathcal{T}) = \alpha(\mathcal{L}(\mathcal{T}))$.

Now suppose $(T, B_t)$ is a tree decomposition of $L$.
Let $S_t$ be the vertices of $G$ adjacent only to edges of $B_t$.
Equivalently, $\mathcal{L}(S_t) \subseteq B_t$
and $S_t$ is maximal with the property.
Let $e=uv$ be an edge of $G$.
The edges of $G$ adjacent to $u$ or $v$ form a clique in $L$,
hence there must be a bag $B_{t_0}$ containing $\mathcal{L}(\{u, v\})$
(by Helly property for trees),
and we conclude that $uv$ is contained in $S_{t_0}$.
Suppose a vertex $u$ of $G$ is contained in two bags $S_r$ and $S_s$.
The bags $B_r$ and $B_s$ contain all edges of $G$ incident with $u$, $E_u$,
hence $E_u$ is contained in all bags $B_t$ on the unique $r\text{-}s$ path in $T$,
and therefore $u$ is contained in the bags $S_t$ along the same path.
We conclude that $(T, S_t)$ is a tree decomposition of $G$
with $\mu$-width at most $\alpha(T, B_t)$.
\end{proof}

Similarly as before, we have that $\mu\text{-}tw(G) \le \alpha\text{-}tw(G)$,
and combining this with the theorem above we get
\[
\alpha \text{-} tw(G) \ge \mu \text{-} tw(G) =
\alpha \text{-} tw(L^2(G)) \ge \mu \text{-} tw(L^2(G)).
\]
For completeness, we note
\[
\alpha(G) \ge \mu(G) = \alpha(L^2(G)) \ge \mu(L^2(G)).
\]

\subsection{Approximating well-behaved width measures}
\label{sec:approx_well_behaved_msr}
The overall strategy is similar to other treewidth approximation algorithms,
in particular to the approximation of fractional hypertree
width~\cite{ftw-approx}.
Suppose $\lambda_H$ is a well-behaved measure throughout.

\begin{thm}
  \label{thm:alpha_separator}
  There exists an $||H||^{O(k)}$-time algorithm that
  on input hypergraph $H$;
  sets of vertices $A, B$;
  and an integer $k$,
  either
  \begin{itemize}
  \item
    finds an $(A, B)$-separator $S$ with $\lambda_H(S) \le {k + 1 \choose 2}k$, or
  \item
    correctly concludes that $A$ and $B$ cannot be separated by a set
    $S$ with $\lambda_H(S) \le k$ or that $\lambda\text{-}tw(H) > k$.
  \end{itemize}
\end{thm}

We need few preliminary lemmas first.
Given a hypergraph $H$ and a decomposition $(T, B_t)$ of $H$,
denote the subtree of $T$ induced by the nodes containing
$v \in V(H)$ by $T_v$.
A set $S$ is an $(A, B)$-separator in a hypergraph $H$
if it is an $(A, B)$-separator in $\underline{H}$,
that is there are no edges between $A \setminus S$ and
$B \setminus S$ and $A \cap B \subseteq S$.

\begin{lemma}
  For every hypergraph $H$ and set $S$ we have
  \[
  \alpha_{\underline{H}}(S) \le \lambda_H(S).
  \]
\end{lemma}
\begin{proof}
  Let $I \subseteq S$ be an independent set of $\underline H$.
  We have $\lambda_H(\{i\}) \ge 1$ for each $i \in I$ by Property 1. of
  $\lambda_H$.
  There are no edges in $\underline{H}$
  between any pair of vertices $i, j \in I$,
  hence by Property 2. of $\lambda_H$ we get
  $\lambda_H(I) = \sum_{i \in I} \lambda_H(\{i\}) \ge |I|$.
  By Property 3. of $\lambda_H$ we see that
  $\lambda_H(I) \le \lambda_H(S)$.
  Now suppose $I \subseteq S$ is maximum-sized.
  We get
  \[
  \alpha_{\underline{H}}(S) = |I| \le \lambda_H(I) \le \lambda_H(S). \qedhere
  \]
\end{proof}

\begin{lemma}
  \label{thm:extending_H}
  Suppose $H$ is a hypergraph with $\lambda\text{-}tw(H) \le k$.
  Let $H_0, H_1, \ldots, H_\ell$ be a sequence
  with $H_0 := H$ and
  $H_{i+1}$ be obtained from $H_i$ by adding an arbitrary edge
  $uv$ such that
  \[
  \lambda_{H_0}(N_{\underline{H_i}}(u) \cap N_{\underline{H_i}}(v)) > k.
  \]
  Then every tree-decomposition of $H$ of $\lambda$-width at most $k$
  is a tree-decomposition of $H_\ell$.
  Furthermore, if $S$ is an $(A, B)$-separator of $H$
  with $\lambda_H(S) \le k$,
  then $S$ is an $(A, B)$-separator of each $H_i$.
\end{lemma}
\begin{proof}
  The proof is by induction on $\ell$.
  The claim is trivial for $\ell=0$.
  Suppose $\ell > 0$ and that the lemma holds for smaller values.
  Let $\mathcal{T} = (T, B_t)$ be an arbitrary tree decomposition of $H$
  such that $\lambda_{H_0}(\mathcal{T}) \le k$.
  By the induction hypothesis, $\mathcal{T}$ is a decomposition of $H_{\ell-1}$
  as well.
  If $\mathcal{T}$ is not a decomposition of $H_\ell$,
  then the edge $uv$ added on the $\ell$-th step is not covered in $\mathcal{T}$,
  hence $T_u$ and $T_v$ are disjoint.
  Let $Z = N_{\underline{H_{\ell-1}}}(u) \cap N_{\underline{H_{\ell-1}}}(v)$ and
  let $x$ be an arbitrary vertex of $Z$.
  We see that $T_x$ intersects both $T_u$ and $T_v$.
  It follows that $x$ is contained in every bag along the
  unique shortest $T_u$--$T_v$ path in $T$, $P_{uv}$.
  Note that $P_{uv}$ is non-empty; it contains at least two nodes of $T$.
  But $x$ was arbitrary,
  so the entire $Z$ is contained in every bag $B_t$,
  where $t \in P_{uv}$.
  This is a contradiction, since $\lambda_{H_0}(Z) > k$
  and $\lambda_{H_0}(\mathcal{T}) \le k$.

  For the second part,
  suppose $S$ is not an $(A, B)$-separator of $H_i$,
  and let $i$ be minimal with this property.
  Clearly $i > 0$.
  Suppose the edge $uv$ is added on the $i$-th step.
  Since $S$ is not a separator of $H_i$, w.l.o.g. we have
  $u \in A \setminus S$ and $v \in B \setminus S$.
  But then every vertex of
  $Z = N_{\underline{H_{i-1}}}(u) \cap N_{\underline{H_{i-1}}}(v)$
  is contained in $S$ in order to prevent an $u$--$v$ path of length $2$ in
  $\underline{H_{i-1}}$,
  hence $Z \subseteq S$ and $k < \lambda_H(Z) \le \lambda_H(S)$.
  We arrive at a contradiction, since
  \[
  \lambda_H(S) \le
  k <
  \lambda_H(Z) \le
  \lambda_H(S). \qedhere
  \]
\end{proof}

\begin{define}
  Let $H$ be a hypergraph and let
  $\underline{\underline{H}}$ be obtained from
  $\underline{H}$ by
  adding edges $uv$ such that
  $\lambda_{H}(N_{\underline{H_i}}(u) \cap N_{\underline{H_i}}(v)) > k$
  as in the previous lemma
  until no longer possible
  and in any order, say Lex-first.
\end{define}

Note that $\underline{\underline{H}}$ can be built from $H$
in $||H||^{O(k)}$ time.
Indeed,
at most $O(n^2)$ new edges are added,
each step inspects at most $O(n^2)$ non-adjacent pairs, and
for each pair $(u, v)$ it can be checked if
$\lambda_{H}(N_{\underline{H_i}}(u) \cap N_{\underline{H_i}}(v)) > k$
in $||H||^{O(k)}$ time,
hence the total running time is $||H||^{O(k)}$.

\begin{lemma}
  \label{thm:extended_H}
  If $H$ is a hypergraph with $\lambda\text{-}tw(H) \le k$,
  then for every clique $C$ of $\underline{\underline{H}}$
  we have $\lambda_H(C) \le k$.
  Furthermore, every $(A, B)$-separator $S$ of $H$ with $\lambda_H(S) \le k$
  is an $(A, B)$-separator in $\underline{\underline{H}}$.
\end{lemma}
\begin{proof}
  Suppose $C$ is a clique in $\underline{\underline{H}}$
  and let $\mathcal{T}$ be a
  minimum $\lambda$-width tree decomposition of $H$.
  Then $\mathcal{T}$ is a decomposition of $\underline{\underline{H}}$
  by Lemma~\ref{thm:extending_H},
  and since $C$ is a clique,
  there must be a bag $B$ of $\mathcal{T}$ containing $C$
  by the Helly property for trees.
  We conclude that
  \[
  \lambda_H(C) \le \lambda_H(B) \le \lambda_H(\mathcal{T})
  = \lambda\text{-}tw(H) \le k.
  \]
  The second part of the statement follows directly from
  Lemma~\ref{thm:extending_H}.
\end{proof}

\begin{lemma}
  \label{thm:connected_subsets}
  \cite{clique_cutsets, tarjan-clique-decomp, ftw-approx}
  \label{clique_cutset}
  Given a graph $G$ it is possible to construct
  in polynomial time a set $\mathcal{K}$ of at
  most $|V(G)|$ subsets of $V(G)$
  such that
  \begin{itemize}
  \item
    each clique $C$ of $G$ is contained in a set $K \in \mathcal{K}$, and
  \item
    $K \setminus C$ is contained in a connected component
    of $G \setminus C$ for each set $K \in \mathcal{K}$
    and every clique $C$ of $G$.
  \end{itemize}
\end{lemma}

\begin{proof}
  [Proof of Theorem~\ref{thm:alpha_separator}]
  Suppose there is an $(A, B)$-separator $S$ with $\lambda_H(S) \le k$
  and that $\lambda\text{-}tw(H) \le k$.
  Let $I$ be a maximum-sized independent subset of $S$
  with respect to $\underline{\underline{H}}$.
  Note that
  $|I| = \alpha_{\underline{\underline{H}}}(S)
  \le \alpha_{\underline{H}}(S) \le \lambda_H(S) \le k$,
  hence $|I| \le k$.
  Let
  \begin{align*}
    X &= \bigcup_{u, v \in I; u \neq v} \left(N_{\underline{\underline{H}}}(u)
    \cap N_{\underline{\underline{H}}}(v)\right)\\
  N_v &= N_{\underline{\underline{H}}}(v) \setminus X \cup \{v\}
  \text{ }\text{ }\text{ }\text{ }\text{ }\text{ }\text{ for } v \in I,
  \text{ and}\\
  Y &= \bigcup_{v \in I} N_v.
  \end{align*}
  It follows that $S \subseteq X \cup Y$
  since $I$ is maximal,
  and that $N_v \cap S$ is a clique of $\underline{\underline{H}}$
  since $I$ is maximum-sized.
  For $v \in I$, let $\mathcal{K}_v$ be the sets of
  $\underline{\underline{H}}[N_v]$
  guaranteed by Lemma~\ref{clique_cutset},
  and pick a $K_v \in \mathcal{K}_v$ to contain $N_v \cap S$ for each $v \in I$.
  Since $K_v \setminus (N_v \cap S) = K_v \setminus S$
  is contained in a component of $\underline{\underline{H}}[N_v] \setminus S$,
  it follows that $K_v \setminus S$
  is contained in a component of $\underline{\underline{H}} \setminus S$.
  By Lemma~\ref{thm:extended_H} we have that $S$ is an
  $(A, B)$-separator of $\underline{\underline{H}}$,
  so every vertex of $K_v \setminus S$ is either reachable from $A$ in
  $\underline{\underline{H}} \setminus S$, or none are.
  Let $J_1 \subseteq I$ be the vertices corresponding
  to the prior case, and $J_2 = I \setminus J_1$ -- the latter.
  In case $K_v \setminus S$ is empty, put $v$ in $J_1$ or $J_2$ arbitrarily.
  Let $Z = X \bigcup (\cup_{v \in I} K_v)$.
  Note that $S \subseteq Z$.
  For $v \in J_2$,
  we call a vertex $u \in K_v$ \emph{bad} if
  there is a path $P_{Au}$ from $A$ to $u$ in
  $\underline{\underline{H}}$
  such that $P_{Au} \cap Z = \{u\}$,
  and for $u \in J_1$
  we call a vertex $u \in K_v$ \emph{bad} if
  there is a path $P_{Bu}$ from $B$ to $u$ in
  $\underline{\underline{H}}$
  such that $P_{Bu} \cap Z = \{u\}$.
  Observe that by the definition of $J_1$ and $J_2$,
  every bad vertex is contained in $S$.
  If $x \in J_1$, $y \in J_2$,
  $u \in K_x$ and $v \in K_y$,
  we call the pair $\{u, v\}$ \emph{bad}
  if there is a path $P_{uv}$ from $u$ to $v$ in $\underline{\underline{H}}$,
  such that $P_{uv} \cap Z = \{u, v\}$.
  If $\{u, v\}$ is a bad pair,
  then $S$ must contain at least one of the two vertices,
  as otherwise $u$ and $v$ are in the same connected
  component of $\underline{\underline{H}} \setminus S$,
  contradicting the choices for $J_1$ and $J_2$.
  To summarise:
  \begin{itemize}
  \item
    $S \subseteq Z$,
  \item
    $S \cap K_v$ is a clique in $\underline{\underline{H}}$ for each $v \in I$,
  \item
    $S$ contains all bad vertices, and
  \item
    $S$ intersects all bad pairs.
  \end{itemize}
  Let $S'$ be another set with these properties.
  We claim that $S' \cup X$ is an $(A, B)$-separator
  in $\underline{\underline{H}}$
  with $\lambda_H(S' \cup X) \le {k + 1 \choose 2}k$.
  Since $X$ is a union of at most ${k \choose 2}$
  coneighbourhoods of non-adjacent vertices in $\underline{\underline{H}}$,
  we have $\lambda_H(X) \le {k \choose 2}k$ by
  the definition of $\underline{\underline{H}}$ and Property 2. of $\lambda_H$.
  Furthermore,
  $\lambda_H(S'\setminus X) \le \sum_{v \in I} \lambda_H(S' \cap K_v) \le k^2$,
  because $S' \cap K_v$ is a clique in $\underline{\underline{H}}$,
  and hence $\lambda_H(S' \cap K_v) \le k$ by Lemma~\ref{thm:extended_H}.
  By putting everything together we conclude that
  $\lambda_H(S' \cup X) \le k^2 + {k \choose 2}k = {k + 1 \choose 2}k$.

  Now suppose $S' \cup X$ is not an $(A, B)$-separator in $\underline{\underline{H}}$,
  and that $P_{a, b}$ is an $a$--$b$ path in $\underline{\underline{H}}$,
  disjoint from $S' \cup X$, where $a \in A$ and $b \in B$.
  Clearly, $P_{a, b}$ is disjoint from $X$,
  so $P_{a, b} \cap Z = P_{a, b} \cap (Z \setminus X)$.
  Let $p_1, \ldots p_m$ be the vertices of $P_{a, b} \cap Z$,
  ordered as the path is traversed from $a$ to $b$.
  Since $S \subseteq Z$ is an $(A, B)$-separator in $\underline{\underline{H}}$,
  $P_{a, b} \cap Z \supseteq P_{a, b} \cap S$ is non-empty.
  We claim that if $p_1 \in K_v$ and $p_m \in K_u$,
  then $v \in J_1$ and $u \in J_2$.
  Indeed, if $v \in J_2$,
  then stripping $P_{a, b}$ from $a$ to $p_1$ yields
  a path $P$ such that $P \cap Z = \{p_1\}$,
  hence $p_1$ is bad and it must be included in $S'$,
  a contradiction.
  Similarly,
  if $u \in J_1$,
  then by denoting the subpath of $P_{a, b}$ from $p_m$ to $b$ by $P$,
  we see that $P \cap Z = \{p_m\}$, hence $p_m$ is bad,
  a contradiction.
  Therefore, there must be a pair $\{p_j, p_{j+1}\}$
  of two consecutive vertices $p_j, p_{j+1}$ of $P_{a,b} \cap Z$,
  such that $p_j \in K_u$ with $u \in J_1$
  and $p_{j+1} \in K_v$ with $v \in J_2$.
  But then for the $p_j$--$p_{j+1}$ subpath, $P$, of $P_{a, b}$,
  we have $P \cap Z = \{p_j, p_{j+1}\}$,
  and hence the pair $\{p_j, p_{j+1}\}$ is bad,
  and therefore it cannot be disjoint from $S'$,
  a contradiction.

  Consider a 2-CNF formula $\psi$
  \begin{itemize}
  \item
    with variables $x_v$ for each vertex $v \in Z \setminus X$,
  \item
    clauses $\neg x_v \lor \neg x_u$ for each pair
    of non-adjacent in $\underline{\underline{H}}$ vertices $v, u$
    contained in the same set $K_w$,
  \item
    with clauses $x_v \lor x_u$ for each bad pair $\{v, u\}$, and
  \item
    set all variables $x_v$ corresponding to bad vertices $v$
    to \emph{true}.
  \end{itemize}
  Observe that $\psi$ has a satisfying assignment -- setting the variables
  corresponding to vertices of
  $S \setminus X$ to \emph{true} and the rest to \emph{false} yields one.
  Conversely, if $S'$ is the set of variables set to \emph{true}
  in a satisfying assignment,
  then by the discussion above we have
  that $S' \cup X$ is an $(A, B)$-separator in $\underline{\underline{H}}$,
  and hence also of $H$,
  and $\lambda_H(S' \cup X) \le {k + 1 \choose 2}k$.

  The algorithm to find an $(A, B)$-separator of bounded $\lambda$-width
  is as follows:
  \begin{enumerate}
  \item
    Construct $\underline{\underline{H}}$,
  \item
    guess $I$,
  \item
    build $X$ and $N_v$ for each $v \in I$,
  \item
    build $\mathcal{K}_v$, guess $K_v \in \mathcal{K}_v$ for each $v \in I$,
    and build $Z$,
  \item
    guess a partition $(J_1, J_2)$ of $I$ ($J_1$ or $J_2$ may be empty),
  \item
    find which vertices are bad, which pairs are bad, and build $\psi$,
  \item
    if $\psi$ does not have a satisfying assignment, reject the branch;
    otherwise find a set $S'$ corresponding to a satisfying assignment,
  \item
    if for any $v \in I$ we have $\lambda_H(S' \cap K_v) > k$,
    declare $\lambda\text{-}tw(H) > k$ and terminate;
    if $S' \cup X$ is not an $(A, B)$-separator in $H$, reject the branch;
    otherwise return $S' \cup X$.
  \end{enumerate}

  As noted before, building $\underline{\underline{H}}$
  takes $||H||^{O(k)}$ time.
  Next, since $|I| \le k$,
  we may try all $O(n^k)$ independent sets of size at most $k$ one by one,
  and thus add a factor of $O(n^k)$ to the total running time.
  Constructing each $\mathcal{K}_v$ is cheap,
  but guessing each $K_v \in \mathcal{K}_v$ adds one more $O(n^k)$ factor.
  There are at most $2^k$ choices for $(J_1, J_2)$,
  so we get an additional factor of $O(2^k)$.
  Step 6. can be done in $n^{O(1)}$ time.
  Since finding a satisfying assignment or correctly concluding that a 2-CNF
  formula is unsatisfiable can be done in linear time with respect its length
  ~\cite{tarjan},
  Step 7. takes $O(n^2)$ time.
  The last step takes $||H||^{O(k)}$ time.
  The total running time is $||H||^{O(k)}$.
\end{proof}

\begin{lemma}
  \label{thm:W_separator_math}
  Suppose $H$ is a hypergraph with $\lambda\text{-}tw(H) \le k$
  and $W \subseteq V(H)$.
  There exists a set $S \subseteq V(H)$ with $\lambda_H(S) \le k$ such that
  for every connected component $C_i$ of $\underline{H} \setminus S$ we
  have $\lambda_H(W \cap C_i) \le \lambda_H(W) / 2$.
\end{lemma}

\begin{proof}
  Let $(T, B_t)$ be a tree decomposition of $H$ of $\lambda$-width at most $k$.
  If $t_1t_2$ is an edge of $T$,
  $T - t_1t_2$ has two connected components, $T_1$ and $T_2$ respectively,
  and let us denote $\cup\{B_t : t \in V(T_i)\}$ by $V_i$.
  We have that $Z := B_{t_1} \cap B_{t_2} = V_1 \cap V_2$
  is a $(V_1, V_2)$-separator.
  If $\lambda_H(W \cap (V_i \setminus Z)) \le \lambda_H(W) / 2$
  for $i \in \{1, 2\}$,
  then we may choose $Z$ as $S$ and we are done.
  Otherwise, orient the edge $t_1t_2$ towards $t_i$,
  where $\lambda_H(W \cap (V_i \setminus Z)) > \lambda_H(W) / 2$.

  Suppose $S$ cannot be chosen as $B_{t_1} \cap B_{t_2}$,
  and all edges are oriented.
  Let $t$ be a node of $T$ with outdegree zero.
  Let $t_1, \ldots, t_m$ be the neighbours of $t$ in $T$,
  so that each edge $t_it$ is oriented from $t_i$ to $t$,
  and let $T_1, \ldots, T_m$ be the corresponding components of $T \setminus t$.
  Define $V_i$ as before.
  Every connected component of $\underline{H} \setminus B_t$
  is contained in some set $V_i \setminus B_t$,
  but from the orientation of $t_it$
  we get that
  \[
  \lambda_H(W \cap (V_i \setminus B_t))
  \le \lambda_H(W \cap (V_i \setminus (B_t \cap B_{t_i})))
  < \lambda_H(W) / 2,
  \]
  so we may choose $B_t$ as $S$.
\end{proof}

\begin{thm}
  \label{thm:W_separator_alg}
  There is an $||H||^{O(r + k)}$-time algorithm that given
  a hypergraph $H$ and a set $W \subseteq V(H)$ with $\lambda_H(W) \le r$
  as input either
  \begin{itemize}
  \item
    finds a partition $(A, B)$ of $W$ and
    an $(A, B)$-separator $S$ of $H$ with $\lambda_H(S) \le {k + 1 \choose 2}k$
    such that $\lambda_H(A \setminus S) \le \frac23r + k$ and
    $\lambda_H(B \setminus S) \le \frac23r + k$, or
  \item
    correctly concludes that $\lambda\text{-}tw(H) > k$.
  \end{itemize}
\end{thm}

\begin{proof}
  Suppose that $\lambda\text{-}tw(H) \le k$,
  let $S$ be the set guaranteed by Lemma~\ref{thm:W_separator_math},
  and let $C_1, \ldots C_t$ be the connected components of
  $\underline{H} \setminus S$,
  such that $\lambda_H(W \cap C_i) \le r/2$,
  and further suppose that $\lambda_H(C_i \cap W) \le \lambda_H(C_{i+1} \cap W)$
  for each $i < t$.
  If $V(H) = S$ there is nothing to do, so suppose $t \ge 1$.
  We seek to find a partition $(\hat{A}, \hat{B})$ of $W \setminus S$
  such that $S$ is an $(\hat{A}, \hat{B})$-separator and
  $\lambda_H(\hat{A}), \lambda_H(\hat{B}) \le \frac23 r$.
  Let $R_j = \cup_{i=1}^j C_i \cap W$
  and suppose $\ell$ is the last index such that $\lambda_H(R_\ell) \le r/2$.
  If $\ell = t$, set $\hat{A} = R_\ell$ and $\hat{B} = \emptyset$.
  If $\ell + 1 = t$, set $\hat{A} = R_\ell$ and $\hat{B} = C_t \cap W$.
  Otherwise let $Q = \cup_{i=\ell+2}^t C_i \cap W$ and note that
  $\lambda_H(Q) \ge \lambda_H(C_{\ell+1} \cap W)$ from the ordering.
  We claim that at least one pair of
  $\{(R_\ell \cup (C_{\ell+1} \cap W), Q), (R_\ell, (C_{\ell+1} \cap W) \cup Q)\}$
  is a possible choice for $(\hat{A}, \hat{B})$.
  Indeed, if $\lambda_H(R_\ell \cup (C_{\ell+1} \cap W)) > \frac23r$,
  then $\frac13 r > \lambda_H(Q) \ge \lambda_H(C_{\ell+1} \cap W)$,
  hence $\lambda_H(Q \cup (C_{\ell+1} \cap W)) < \frac23r$.

  Let $(A, B)$ be an arbitrary partition of $W$ such that
  $\hat A \subseteq A$ and $\hat B \subseteq B$.
  The desired separator in the first condition of the statement
  can be found with the algorithm in Theorem~\ref{thm:alpha_separator}
  with arguments $A$ and $B$ since
  $\lambda_H(A) \le \frac 23r + k$ and $\lambda_H(B) \le \frac 23r + k$.
  It remains to show that a partition $(A, B)$ can be found efficiently.
  Let $I$ be a maximal independent subset of $\hat A$ in $\underline{H}$,
  and let $\Gamma(I)$ be
  $I \cup \{v \in V(H) : u \in I, uv \in E(\underline{H})\}$.
  We see that
  $\hat{A} \subseteq \Gamma(I) \cap W \subseteq \hat{A} \cup S$,
  and $\hat{B} \subseteq W \setminus \Gamma(I) \subseteq \hat{B} \cup S$,
  so $(\Gamma(I) \cap W, W \setminus \Gamma(I))$
  is a possible choice for $(A, B)$.
  Furthermore, $|I| \le \lambda_H(\hat A) \le \frac23r$,
  so $I$ can be found by trying all stable sets of size at most $\frac23r$
  in time $O(n^{\frac23r})$.
  The final algorithm is:
  \begin{enumerate}
  \item
    Guess $I$, and independent set of $\underline{H}$ with $|I| \le \frac23r$,
  \item
    construct $A = \Gamma(I) \cap W$
    and $B = W \setminus \Gamma(I)$,
  \item
    reject the branch
    if $\lambda_H(A) > \frac 23r + k$ or $\lambda_H(B) > \frac 23r + k$,
  \item
    run the algorithm in Theorem~\ref{thm:alpha_separator}
    with arguments $H$, $A$ and $B$.
  \end{enumerate}
  The total running time is $||H||^{O(r+k)}$.
\end{proof}

\begin{thm}
  \label{thm:approx_decomp}
  There exists a $||H||^{O(k^3)}$-time algorithm
  that given
  a hypergraph $H$,
  an integer $k$
  and a set $W \subseteq V(H)$ with $\lambda_H(W) \le \frac32(k^3 + k^2) + 3k + 3$,
  either
  \begin{itemize}
  \item
    finds a tree decomposition, $\mathcal{T} = (T, B_t)$, of $H$,
    such that $\lambda(\mathcal{T}) \le 2k^3 + 2k^2 + 3k + 3$
    and $W \subseteq B_t$ for some $t \in V(T)$, or
  \item
    correctly concludes that $\lambda\text{-}tw(H) > k$.
  \end{itemize}
\end{thm}
\begin{proof}
  Let $K = \frac32(k^3 + k^2) + 3k + 3$ and $s = \frac 12 (k^3 + k^2)$.
  The algorithm starts by creating a set $W^*$ such
  that $W \subseteq W^*$ and $\lambda_H(W^*) = \min\{K, \lambda_H(V(H))\}$.
  Let $S$ and $(A_1, A_2)$ be the separator
  and the partition of $W^*$ from Theorem~\ref{thm:W_separator_alg};
  let $V_1$ be the vertices of $\underline{H} \setminus S$ contained in
  connected components with a non-empty intersection with $A_1$
  and let $V_2$ be $V(H) \setminus (V_1 \cup S)$.
  For $i \in \{1, 2\}$, the algorithm checks if $V_i \not\subseteq A_i$,
  in which case it makes a recursive call with arguments
  $H_i = H[V_i \cup S]$ and $W_i = A_i \cup S$
  to create a tree decomposition $(T_i, B_{t,i})$
  or conclude that $\lambda\text{-}tw(H) > k$ if $\lambda\text{-}tw(H_i) > k$.
  In case $V_i \subseteq A_i$, let $T_i$ be the empty tree.
  After $T_1$ and $T_2$ have been created,
  the corresponding node containing $W_i$ in each tree is joined
  to a new node with bag $W^* \cup S$.
  If either $T_1$ or $T_2$ is empty, ignore it.

  If $\lambda\text{-}tw(H) \le k$, then
  $\lambda\text{-}tw(H_i) \le k$ and
  \[
  \lambda(W_i) \le \lambda(A_i \setminus S) + \lambda(S) \le
  \frac 23 K + k + s = K - 1,
  \]
  so we can prove by induction that the algorithm creates
  tree decompositions $(T_1, B_{t,1})$ of $H_1$ and $(T_2, B_{t,2})$ of $H_2$
  of width at most $K + s$.
  If $\lambda\text{-}tw(H) > k$,
  the algorithm will either create a decomposition of width at most $K + s$,
  or find an induced subhypergraph $H'$ with $\lambda\text{-}tw(H') > k$,
  implying $\lambda\text{-}tw(H) > k$.

  The total number of steps can be bounded by the number of recursive calls
  times the maximum number of steps for each call.
  The number of steps in each call is clearly $||H||^{O(k^3)}$.
  We claim that the number of calls is bounded by $|V(H) \setminus W^*| + 1$.
  The proof is by induction on $|V(H) \setminus W^*|$.
  If $|V(H) \setminus W^*| = 0$, then $\lambda_H(V(H)) \le K$,
  hence the algorithm creates one bag and does not make recursive calls.

  Now suppose $V(H) \not\equiv W^*$.
  We claim that the number of recursive calls to create $T_i$
  is bounded by $|V_i \setminus W^*|$.
  Indeed, if $|V_i \setminus W^*| = |V_i \setminus A_i| = 0$,
  then $T_i$ is empty, and there are no recursive calls.
  If $|V_i \setminus W^*| > 0$,
  then by the induction hypothesis,
  there are at most $|V(H_i) \setminus W^*_i| + 1$ calls.
  Since $\lambda(W_i) \le K-1$,
  either $W^*_i \equiv V(H_i)$ or $\lambda(W_i^*) = K$,
  and hence $W_i^*$ contains a vertex of $V_i \setminus W^*$,
  so in either case $|V(H_i) \setminus W^*_i| + 1 \le |V_i \setminus W^*|$.
  Therefore the total number of calls is bounded by
  $|V_1 \setminus W^*| + |V_2 \setminus W^*| + 1 \le |V \setminus W^*| + 1$.

  Overall, the total number of recursive calls,
  and hence the number of bags crated,
  is bounded by $|V(H)|$,
  and therefore the total running time is $||H||^{O(k^3)}$.
\end{proof}

Theorem~\ref{thm:approx_well_behaved} is the special case
of Theorem~\ref{thm:approx_decomp} when $W = \emptyset$.

\section{Open problems and further work}
It is not clear whether the restriction on the rank
is necessary.
We formulate the following conjectures:
\begin{conj}
  There exists a function $f$ such that for every clutter $H$ with $m$ edges
  we have
  \[
  |b(H)| \le m^{f(\mu(H))}.
  \]
\end{conj}

\begin{conj}
  There exists a function $f$ such that for every clutter $H$ with $m$ edges
  and any subset $S$ of its vertices
  we have
  \[
  |tr_S(b(H))| \le m^{f(\mu_H(S))}.
  \]
\end{conj}

It is also interesting to see if the approximation of $\mu\text{-}tw$ in \S~4
can be extended to clutters of any rank,
and if there
are efficient reductions
of $\alpha\text{-}tw$ to $\mu\text{-}tw$ for clutters
and vice-versa.

Another question is to find more examples of problems
with solutions that can be read from the blocker.
Indeed, $b(H)$ contains a lot of important structural
information that is difficult to access otherwise.
Finding an explicit family of problems of this type
can lead to further developments of $\mu\text{-}tw$.

\bibliographystyle{alpha}
\bibliography{treestab}

\end{document}